\documentclass[12pt]{amsart}

\usepackage{amssymb}
\usepackage{amsmath}
\usepackage{amsthm}
\usepackage{amscd}
\usepackage[margin=1.1in]{geometry}
\usepackage{hyperref}
\usepackage{appendix}

\theoremstyle{plain}
\newtheorem{theorem}{Theorem}[section]
\newtheorem{proposition}{Proposition}[section]
\newtheorem{corollary}{Corollary}[section]
\newtheorem{lemma}{Lemma}[section]

\theoremstyle{remark}

\newtheorem{remark}{Remark}[section]

\DeclareMathOperator{\lin}{span}
\DeclareMathOperator{\diverg}{div}
\DeclareMathOperator{\Hess}{Hess}
\DeclareMathOperator{\tr}{Tr}

\begin{document}

\title{Dirac and magnetic Schr\"odinger operators on fractals}
\author{Michael Hinz$^1$}
\address{Mathematisches Institut, Friedrich-Schiller-Universit\"at Jena, Ernst-Abbe-Platz 2, 07737, Germany and Department of Mathematics,
University of Connecticut,
Storrs, CT 06269-3009 USA}
\email{Michael.Hinz.1@uni-jena.de and Michael.Hinz@uconn.edu}
\thanks{$^1$Research supported in part by NSF grant DMS-0505622 and by the Alexander von Humboldt Foundation (Feodor Lynen Research Fellowship Program)}
\author{Alexander Teplyaev$^2$}
\address{Department of Mathematics, University of Connecticut, Storrs, CT 06269-3009 USA}
\email{Alexander.Teplyaev@uconn.edu}
\thanks{$^2$Research supported in part by NSF grant DMS-0505622}
\date{\today}

\begin{abstract}
In this paper we define (local) Dirac operators and magnetic Schr\"odinger Hamiltonians on fractals and prove their (essential) self-adjointness. To do so we use the concept of $1$-forms and derivations associated with Dirichlet forms as introduced by Cipriani and Sauvageot, and further studied by the authors jointly with R\"ockner, Ionescu and Rogers. For simplicity our definitions and results are formulated for the Sierpinski gasket with its standard self-similar energy form. We point out how they may be generalized to  other spaces, such as the classical   Sierpinski carpet. 
\tableofcontents
\end{abstract}
\maketitle

\section{Introduction}

The aim of the present paper is to introduce natural (local) Dirac and magnetic Schr\"odinger operators on fractal spaces  and to prove that they are (essentially) self-adjoint. Our analysis uses the concept of $1$-forms in the content of the Dirichlet form theory, and is  based on recent results 
on $1$-forms and vector fields  \cite{CGIS11,CS03,CS09,IRT} and \cite{HRT}, respectively. 
 
 To make the paper more accessible and to approach the most interesting classical examples, we formulate our definitions and results for the Sierpinski gasket, and later provide some remarks how to modify them for more general   fractals and  other spaces. 
 This is particularly straightforward to do for species with a resistance form, in the sense of Kigami \cite{Ki03,Ki08,Ki12} (see also \cite{BBKT,Hino10,T08}), such as the classical two-dimensional Sierpinski carpet, but many results are valid for much more general spaces. In particular, extending our results for spaces that are not locally compact will be subject of future work.

Our space of $1$-forms is  a Hilbert space, which allows to identify $1$-forms and vector fields, and to introduce other notions of vector analysis, as recently done in \cite{HRT} (which generalizes earlier approaches to vector analysis on fractals, see \cite{Ki93,Ki08,Ku93,Str00,T08}). This is a part of a  comprehensive program to study  vector equations on general non-smooth spaces   which carry a diffusion process or, equivalently, a local regular Dirichlet form. 

\def\arti{and references therein}

The study of the  Laplacian on fractal graphs was originated   in physics literature (see \cite{Al,B92,DABK,GAM,R,RT}), and for a selection of recent mathematical physics results see \cite[\arti]{ADT,ADT2010PRL,FKhStr,GRS,HK-2010,KKPSS}. 
Among the problems  where fractal spaces seem to appear naturally we would like to mention, in particular, the spaces of fractional dimension appearing in quantum gravity  \cite[\arti]{AJL,LR2007,Reuter2011}. 
Besides that, our motivation is  coming 
from the theory of quantum graphs \cite[\arti]{EP,  Exner,  FKW07,  Har08,  KS99,  KS00,  KS03,  Ku04,  Ku05,  KuchmentPost,  Po09}; 
from the spectral theory on fractals \cite[\arti]{eigenpapers,   DerfelEtAl2012,  HSTZ,  KL1,  KL2,  LL,  Mal93,  Mal95,  MT03, RoTe, RomeoSteinhurst}; 
form some questions of non-commutative analysis \cite[\arti]{CL,  CIL, CGIS11,CS03,CS09, IRT} and 
the 
theory of spectral zeta functions \cite{ZetaBook,LvF,ST-merozeta,T-zeta}; 
and from the localization problems \cite[\arti]{AM, MT95, Quint, SA, T98}.

Recall that roughly speaking, the Dirac operator is defined as the square root of the Laplace operator. (Note, however, that classical Dirac operator for diffusions   is a local operator, which excludes the possibility of using the spectral theorem to define it.) 
Depending on context and purpose it appears in various formulations with possibly different complexity and sign conventions. On the real line $D=-i d/dx$ may for instance be regarded as the Dirac operator. Given a Riemannian manifold $M$, its tangent bundle $\Lambda T^\ast M$ can be turned into a Clifford module, and the associated Dirac operator is defined as $D=d+d^\ast$, where $d$ is the exterior derivative and $d^\ast$ is its adjoint, cf. \cite{BGV}. For a spin $1/2$-particle in the plane the Dirac operator is given by $D=-i\sigma_x \partial/\partial x_1-i\sigma_y\partial/\partial x_2$, where $\sigma_x$ and $\sigma_y$ are the respective Pauli spin matrices, \cite{F}. More generally, it may be defined for spinor bundles over spin manifolds, see \cite{BGV, F} or \cite{G} for background and details.

Dirac operators on discrete graphs have for instance been considered in \cite[Section 4]{D93} with a strong emphasis on connections to noncommutative geometry. The paper \cite{Re02} follows a similar spirit and considers related spectral triples and Connes metrics. More recently Dirac operators on discrete graphs and related index theorems have been studied in \cite{Po09}. In this reference they act on a tensor product of form $H_0\oplus H_1$, where $H_0$ and $H_1$ are Hilbert spaces of functions on the vertices and edges, respectively. Roughly speaking, the discrete difference operator $d:H_0\to H_1$ plays the role of the exterior derivative. Denoting its adjoint by $d^\ast:H_1\to H_0$ , the associated Dirac operator is then defined on $H_0\oplus H_1$ by
\begin{equation}\label{E:DiracIntro}
D=\left(
\begin{array}{rr}
0 &d^\ast\\
d &0
\end{array}\right),
\end{equation}
and as a consequence $D^2$ yields the matrix Laplacian acting on $H_0\oplus H_1$. In a somewhat similar fashion \cite{Po09} also investigates Dirac operators and index theorems on quantum graphs (often referred to as metric graphs or quantum wires, \cite{KS99, KS00, Ku04, Ku05}), now within the context of suitable Sobolev spaces. A preceding article dealing with index theorems on quantum graphs is \cite{FKW07}, and a related much earlier reference for Dirac operators is \cite{BT89}. Different quantization schemes are reviewed in \cite{Har08}. 

There is an extensive literature on magnetic Schr\"odinger operators. In the simplest cases (such as for bounded and sufficiently integrable potentials $A$ and $V$) the essential self-adjointness of magnetic Schr\"odinger Hamiltonians
\begin{equation}\label{E:prototype}
(-i\nabla-A)^2+V
\end{equation}
on Euclidean space can be deduced from classical perturbation theorems in Hilbert spaces, cf. \cite[Section X.2]{RS}. More sophisticated pointwise methods can be found in \cite[Section X.3 and X.4]{RS}. Essential self-adjointness results for operators on manifolds may be found in the comprehensive paper \cite{Sh01}, see also \cite{GK12} for singular potentials. Discrete magnetic Schr\"odinger operators on lattices and graphs have for instance been discussed in \cite{BKP85, B92, CTT10b, DM06, HS99, Sh94, Sun94}. The paper \cite{Sh94} introduces discrete magnetic Laplacians on the two-dimensional integer lattice, proves that they have no point spectrum and compares them to almost Mathieu operators and to one-dimensional quasi-periodic operators. This is closely connected to the (one-dimensional) ten martini problem, solved in \cite{AJ}. Periodic magnetic Schr\"odinger operators on the two-dimensional integer lattice are treated in \cite{Sun94} and their spectra, typically of band or Cantor type, are studied. In \cite{HS99} magnetic Schr\"odinger operators on graphs are considered. Under some conditions the analyticity of the bottom of their spectra is verified and relations to corresponding operators on a quotient graph (by a suitable automorphism group) are discussed. The paper \cite{DM06}  also investigates discrete magnetic Laplacians and Schr\"odinger operators on graphs, compares their spectra and heat kernels to the original graph Laplacians, defines related Novikov-Shubin invariants and establishes a long term decay result for the heat kernel trace. In \cite{CTT10b} the essential self-adjointness of a discrete version of (\ref{E:prototype}) is shown, based on a previous result \cite{CTT10a} for operators with zero magnetic potential. Reference \cite{CTT10b} also discusses gauge invariance in terms of holonomy maps. First steps towards magnetic Schr\"odinger operators on fractals had been taken in \cite{BKP85} and \cite{B92} by studying them on infinite Sierpinski lattices. Some decimation techniques for the spectrum and related numerical experiments can be found in \cite{BKP85}. The paper \cite{B92} sets up a renormalization group equation for the magnetic Laplacian and discusses relations to superconductivity. Another branch of literature concerns quantum graphs, see \cite{KS99, KS00, Ku04, Ku05}. The paper \cite{KS03} introduces magnetic Laplacians on metric graphs and, based on results in \cite{KS00}, provides a matrix criterion for the boundary conditions to characterize self-adjointness. In \cite{BGP} 
a metric graph point of view is used to provide a comprehensive study of the two-dimensional periodic square graph lattice with magnetic fields.
The paper \cite{EP} shows that any self-adjoint vertex coupling on a metric graph can be approximated by a sequence of magnetic Schr\"odinger  operators on a network of shrinking tubular neighborhoods.

For prototype examples of fractal sets carrying a diffusion not even the forms of a Dirac operator and a magnetic Laplacian had been clear. This is due to the fact that Laplace operators had been studied in several papers and books (see for instance \cite{Ba, Ki93, Ki01, Ki03, Li, ST-BE} and the references therein), but definitions and results concerning analogs of first order differential operators (gradients) were sparse \cite{Ki93, Ku89, Str00, T00}, and hardly flexible enough to fit a sufficient functional analytic context. 

In \cite{CS03} and \cite{CS09} differential $1$-forms and derivations based on Dirichlet forms had been introduced. In these papers a Hilbert space $\mathcal{H}$ of $1$-forms is constructed as, roughly speaking, the completion of the tensor product $\mathcal{F}\otimes\mathcal{F}$ of the space $\mathcal{F}$ of energy finite functions, a concept that leads to an $L_2$-theory, see for instance \cite{HT} for further explanations. This approach has been studied further in \cite{CGIS11, IRT} and also in \cite{HRT}, where related notions of vector analysis are proposed. In this context the desired objects can be defined. More precisely, our Theorems \ref{T:Dirac} and \ref{T:main} below state that under some conditions, analogs of (\ref{E:DiracIntro}) and (\ref{E:prototype}) define essentially self-adjoint operators on suitable Hilbert spaces of functions and vector fields on fractals, respectively. 

In the next section we review the approach of \cite{CS03, CS09} to $1$-forms based on energy for the specific example of the Sierpinski gasket $K$ with its standard self-similar energy form $(\mathcal{E},\mathcal{F})$.  We recall the definitions of the gradient and divergence operators from \cite{HRT} and the energy Laplacian for functions. In several places we provide auxiliary formulations in terms of harmonic coordinates. In Section \ref{S:Dirac} we define a related Dirac operator $D$ that acts on the tensor product $L_{2,\mathbb{C}}(K,\nu)\otimes \mathcal{H}_\mathbb{C}$ of the spaces of complex-valued square integrable functions (with respect to the Kusuoka measure $\nu$) and complex-valued vector fields on $K$. Theorem \ref{T:Dirac} proves it it is self-adjoint. In Section \ref{S:magnetic} we first provide a priori estimates necessary to introduce a bilinear form $\mathcal{E}^{a,V}$ associated with a magnetic Schr\"odinger Hamiltonian $H^{a,V}$ on $K$. Then we establish a result on its essential self-adjointness on $L_{2,\mathbb{C}}(K,\nu)$, Theorem \ref{T:main}, which merely follows from our definitions, preliminary estimates and a simple KLMN theorem. Finally we prove sort of a gauge invariance result, Theorem \ref{T:gauge}.
Section \ref{S:other} contains some instructions how to generalize the presented results to arbitrary finitely ramified fractals carrying a regular resistance form.

In this paper we generally intend to provide a basic setup to study Dirac operators and magnetic fields on fractals. We do not discuss questions regarding the spectrum, refined pointwise statements or approximations. These topics will be addressed elsewhere.

To simplify notation, sequences or families indexed by the naturals will be written with index set suppressed, e.g. $(a_n)_n$ stands for $(a_n)_{n\in\mathbb{N}}$. Similarly, $\lim_n a_n$ abbreviates $\lim_{n\to \infty}a_n$.

\section{Vector analysis on the Sierpinski gasket}\label{S:vector}

This section recalls a few items of the concept of $1$-forms and vector fields based on Dirichlet forms as studied in \cite{CGIS11,CS03,CS09,IRT} and \cite{HRT}, respectively. For simplicity we formulate definitions and results for the Sierpinski gasket, some comments on finitely ramified fractals are provided in Section \ref{S:other}. For investigations of other physical models on the Sierpinski gasket see for instance \cite{FKhStr, Str09}.

Let $\left\lbrace p_1,p_2,p_3\right\rbrace$ be the vertex set of an equilateral triangle in $\mathbb{R}^2$. The \emph{Sierpinski gasket} $K$ is the unique nonempty compact subset of $\mathbb{R}^2$ such satisfying the self-similarity relation
\[K=\bigcup_{i=1}^3\varphi_i K,\]
where $\varphi_i(x)=x/2+p_i/2$. For our purposes the embedding in $\mathbb{R}^2$ is inessential, only the associated post critically finite self-similar structure $(K,\left\lbrace 1,2,3\right\rbrace, \left\lbrace \varphi_1,\varphi_2,\varphi_3\right\rbrace)$ in the sense of \cite{Ki01} matters. Let $(\mathcal{E},\mathcal{F})$ denote the standard self-similar energy form on $K$, obtained as the increasing limit
\[\mathcal{E}(u):=\lim_n\mathcal{E}_n(u)\]
of a sequence of rescaled discrete energies
\[\mathcal{E}_n(u)=\left(\frac{5}{3}\right)^n\sum_{x\sim_n y}(u(x)-u(y))^2\]
on approximating graphs. For precise definitions and background we refer to \cite{BBST, CS09, Ki93, Ki01}. The form $(\mathcal{E},\mathcal{F})$ is a regular resistance form in the sense of \cite{Ki12}. In particular, $\mathcal{E}^{1/2}$ is a Hilbert norm on the space  $\mathcal{F}/\sim$ obtained from $\mathcal{F}$ by factoring out constants, and there is some constant $c>0$ such that
\begin{equation}\label{E:resistest}
\left\|f\right\|_\infty\leq c\:\mathcal{E}(f)^{1/2}
\end{equation}
for all $f\in\mathcal{F}/\sim$. The space $\mathcal{F}$ is a dense subalgebra of $C(K)$, and in particular
\begin{equation}\label{E:pointwisemult}
\mathcal{E}(fg)^{1/2}\leq \mathcal{E}(f)^{1/2}\left\|g\right\|_{L_\infty(K,\nu)}+\left\|f\right\|_{L_\infty(K,\nu)}\mathcal{E}(g)^{1/2}, \ \ f,g\in\mathcal{F}.
\end{equation}
For any $f\in\mathcal{F}$ we can define a unique (nonatomic) Borel \emph{energy measure} $\nu_f$ on $K$, see for instance \cite{T00}, and polarization yields mutual energy measures $\nu_{f,g}$ for $f,g\in\mathcal{F}$.

The space of nonconstant harmonic functions on $K$ is two dimensional. Let $\left\lbrace h_1, h_2\right\rbrace$ be a complete energy orthonormal system for it. The \emph{Kusuoka energy measure} $\nu$  is defined by 
\[\nu:=\nu_{h_1}+\nu_{h_2},\]
and this definition does not depend on the choice of the complete energy orthonormal system $\left\lbrace h_1, h_2\right\rbrace$. By construction all energy measures $\nu_{f,g}$ are absolutely continuous with respect to $\nu$ and have integrable densities $\Gamma(f,g)$. In particular, we can find Borel versions $Z_{ij}$ of the functions $\Gamma(h_i,h_j)$ and a Borel set $K_0\subset K$ such that for any $x\in K_0$, the real $(2\times 2)$-matrix
\[Z_x:=(Z_{ij}(x))_{ij=1,2}\]
is symmetric, nonnegative definite and of rank one, and we have $\nu(K\setminus K_0)=0$. For $x\in K\setminus K_0$ we may define $Z_x$ to be the zero matrix. See for instance \cite{Hino10, Ku89, T08}. Note that every $Z_x$, $x\in K$, acts as a projection in $\mathbb{R}^2$, and for fixed $x$ the space $(\mathbb{R}^2/ker Z_x, \left\langle \cdot, Z_x\cdot\right\rangle_{\mathbb{R}^2})$ is isometrically isomorphic to the image space $(Z_x(\mathbb{R}^2), \left\langle \cdot, Z_x\cdot\right\rangle_{\mathbb{R}^2})$. In addition, we may assume $K_0$ is such that $h(x)\neq 0$ for all $x\in K_0$.

According to \cite[Theorem 2.4.1]{Ki01} the regular resistance form $(\mathcal{E},\mathcal{F})$ defines a local regular Dirichlet form on $L_2(K,\nu)$. Therefore the measures $\nu_f$, $f\in\mathcal{F}$, coincide with the energy measures in the sense of \cite{FOT94}, and the operation $(f,g)\mapsto \Gamma(f,g)$ taking two members $f,g\in\mathcal{F}$ into the density $\Gamma(f,g)$ coincides with the \emph{carr\'e du champ} in the sense of \cite{BH91}. 

\begin{remark}
Here we consider the $L_2$-space $L_2(K,\nu)$ with respect to the Kusuoka measure. Note that the energy measures $\nu_{f,g}$ are singular with respect to the naturally associated renormalized self-similar Hausdorff measure on $K$, see \cite{BBST, Hino03, Hino05}.  
\end{remark}

Setting 
\[h(x):=(h_1(x),h_2(x)) \ \ \text{ and }\ \ y:=h(x)\] 
we obtain a homeomorphism $h$ from $K$ onto its image $h(K)$ in $\mathbb{R}^2$, and the latter may be viewed with coordinates $y$. The collection of functions of form $f=F\circ h$ with $F\in C^1(\mathbb{R}^2)$ is dense in $\mathcal{F}$, and for any such function the \emph{Kusuoka-Kigami formula} 
\begin{equation}\label{E:Kigami}
\mathcal{E}(f)=\int_K | Z\nabla F(y)|_{\mathbb{R}^2}^2 d\nu
\end{equation}
holds, where  $\nabla F$ is the usual gradient of $F$ in $\mathbb{R}^2$. More generally, by the chain rule \cite[Theorem 3.2.2]{FOT94} the energy measure $\nu_f$ of such $f$ is given by $|Z\nabla F(y)|_{\mathbb{R}^2}d\nu$.

By $L_{2,\mathbb{C}}(X,\nu)$ we denote the natural complexification $L_2(X,\nu)+iL_2(X,\nu)$ of $L_2(X,\nu)$. The closed form $\mathcal{E}$ on $L_2(X,\nu)$ can be complexified by setting 
\begin{equation}\label{E:complexenergy}
\mathcal{E}(f,g):=\mathcal{E}(f_1,g_1)-i\mathcal{E}(f_1,g_2)+i\mathcal{E}(f_2,g_1)+\mathcal{E}(g_1,g_2)
\end{equation}
for any $f=f_1+if_2$ and $g=g_1+ig_2$ from $\mathcal{F}_\mathbb{C}:=\mathcal{F}+i\mathcal{F}$. This yields a positive definite quadratic form $\mathcal{E}$ on $L_{2,\mathbb{C}}(X,\nu)$. That is, $\mathcal{E}$ is conjugate symmetric, linear in the first argument, and $\mathcal{E}(f)\geq 0$ for any $f\in\mathcal{F}_\mathbb{C}$. We will use a similar terminology for the mappings considered in what follows. The form $\mathcal{E}$ is densely defined and closed. Similarly, and in a way consistent with (\ref{E:complexenergy}), also the energy measure $\nu_{f,g}$ and their densities $\Gamma(f,g)$ can be complexified.

Consider $\mathcal{F}_\mathbb{C}\otimes\mathcal{F}_\mathbb{C}$ endowed with the symmetric bilinear form
\begin{equation}\label{E:scalarprod}
\left\langle a\otimes b, c\otimes d\right\rangle_\mathcal{H}=\int_K b\overline{d}\:\Gamma(a,c)d\nu,
\end{equation}
$a\otimes b, c\otimes d\in\mathcal{F}_\mathbb{C}\otimes\mathcal{F}_\mathbb{C}$. Let $\mathcal{H}_\mathbb{C}$ denote the Hilbert space obtained by factoring out trivial elements and completing. Following \cite{CS03, CS09} we refer to it as the \emph{space of $1$-forms} on $K$. For simplicity we will not distinguish between an element $a\otimes b$ and its equivalence class in $\mathcal{H}_\mathbb{C}$.

To rewrite several items in coordinates we also define the space
\[\mathcal{S}_\mathbb{C}:=\lin\left\lbrace f\otimes g: f=F\circ h, g=G\circ h\ \text{ with $F,G\in C^1_\mathbb{C}(\mathbb{R}^2)$}\right\rbrace.\]

\begin{theorem}
The space $\mathcal{S}_\mathbb{C}$ is dense in $\mathcal{H}_\mathbb{C}$.
\end{theorem}
\begin{proof}
Note first that the collection $\widetilde{\mathcal{S}}_\mathbb{C}$ of elements $f\otimes g$ with $f=F\circ h$, $F\in C^1_\mathbb{C}(\mathbb{R}^2)$ and $g\in\mathcal{F}_\mathbb{C}$ is dense in $\mathcal{H}_\mathbb{C}$: By the definition of $\mathcal{H}_\mathbb{C}$ it suffices to approximate finite linear combinations $\sum_i a_i\otimes b_i\in\mathcal{F}_\mathbb{C}\otimes\mathcal{F}_\mathbb{C}$ by elements of $\widetilde{\mathcal{S}}_\mathbb{C}$. For fixed $i$ let $(f_i^{(n)})_n$ be a sequence $\mathcal{E}$-converging to $a_i$. Then
\[\left\|\sum_i a_i\otimes b_i-\sum_i f_i^{(n)}\otimes b_i\right\|_\mathcal{H}^2=\sum_{ij}\int_K b_i\overline{b_j}\Gamma(a_i-f_i^{(n)})d\nu,\]
which converges to zero by the boundedness of the functions $b_i$. On the other hand every element $\sum_i a_i\otimes b_i$ of $\widetilde{\mathcal{S}}_\mathbb{C}$ can be approximated by elements of $\mathcal{S}_\mathbb{C}$: For fixed $i$ let $(g_i^{(n)})_n$ $\mathcal{E}$- converge to $b_i$. The estimate (\ref{E:resistest}) implies uniform convergence, and therefore also
\[\left\|\sum_i a_i\otimes b_i - \sum_i a_i\otimes g_i^{(n)}\right\|_\mathcal{H}^2=\sum_{ij}\int_K(b_i-g_i^{(n)})^2\:\Gamma(a_i,a_j)d\nu\]
goes to zero.
\end{proof}

Recall that we use the coordinate notation $y=y(x)=(h_1(x), h_2(x))$.
\begin{theorem}\label{T:fibers}
There are a family of Hilbert spaces $\left\lbrace \mathcal{H}_{\mathbb{C},x}\right\rbrace_{x\in X}$ and surjective linear maps
$\omega\mapsto \omega_x$ from $\mathcal{H}_\mathbb{C}$ onto $\mathcal{H}_{\mathbb{C},x}$ such that the direct integral
$\int^\oplus_K \mathcal{H}_{\mathbb{C},x}\nu(dx)$ is isometrically isomorphic to $\mathcal{H}_\mathbb{C}$ and in particular,
\[\left\|\omega\right\|_{\mathcal{H}}^2=\int_K\left\|\omega_x\right\|_{\mathcal{H},x}^2\nu(dx), \ \ \omega\in\mathcal{H}_\mathbb{C}.\]
For $\nu$-a.e. $x\in X$ the fiber $\mathcal{H}_{\mathbb{C},x}$ is isomorphic to $\mathbb{C}^2/ker\:Z_x$, and for any $f=F\circ h$ and $g=G\circ h$ with $F,G\in C^1_\mathbb{C}(\mathbb{R}^2)$ we have
\begin{equation}\label{E:Kigami2}
\left\|f\otimes g\right\|_\mathcal{H}^2=\int_K |Z G(y)\nabla F(y)|_{\mathbb{C}^2}^2 d\nu.
\end{equation}
\end{theorem}

Identity (\ref{E:Kigami2}) obviously extends (\ref{E:Kigami}). 

\begin{proof}
For fixed $x\in K_0$ define  a linear map $\varphi_x:\mathcal{S}_\mathbb{C}\to \mathbb{C}^2/ker\:Z_x$ by
\[\varphi_x(g\otimes f):=Z_xG(y)\nabla F(y).\]
Since $h(x)\neq 0$ and linear functions belong to $C^1_\mathbb{C}(\mathbb{R}^2)$, the linear map $\varphi_x$ is surjective. Consequently there exists an isomorphism $\Phi_x$ from $\mathcal{S}_\mathbb{C}/ker\:\varphi_x$ onto $\mathbb{C}^2/ker\:Z_x$, given by
\begin{equation}\label{E:fiberiso}
\Phi_x((f\otimes g)_x)=Z_xG(y)\nabla F(y),
\end{equation}
where $(f\otimes g)_x=\kappa_x(f\otimes g)$ denotes the equivalence class mod $ker\:\phi_x$ of $f\otimes g$ and $\kappa_x$ the canonical epimorphism. From (\ref{E:fiberiso}) we obtain 
\begin{equation}\label{E:pullout}
(f\otimes g)_x=g(x)(f \otimes\mathbf{1})_x
\end{equation}
for any $f,g\in\mathcal{F}_\mathbb{C}$ and $x\in K_0$. We write 
\[\mathcal{H}_{\mathbb{C},x}:=\mathcal{S}/ker\:\varphi_x\] 
and endow this space with the norm
\begin{equation}\label{E:normdef}
\left\|(f\otimes g)_x\right\|_{\mathcal{H},x}:=|Z_xG(y)\nabla F(y)|_{\mathbb{C}^2}.
\end{equation}
Then $\Phi_x$ becomes a isometric isomorphism between Hilbert spaces. For $x\in K\setminus K_0$ we set $\mathcal{H}_{\mathbb{C},x}=\left\lbrace 0\right\rbrace$ and $\kappa_x:=0$. For every $x\in X$ the fiber $\mathcal{H}_{\mathbb{C},x}$ is finite dimensional and therefore $\kappa_x$ extends uniquely to a surjective bounded linear map $\kappa_x:\mathcal{H}_\mathbb{C}\to \mathcal{H}_{\mathbb{C},x}$. For $f\otimes g\in\mathcal{S}_\mathbb{C}$ as above we have
\[\left\|f\otimes g\right\|_\mathcal{H}^2=\int_K|g|^2\Gamma(f)d\nu=\int_K|G(y)|^2|Z_x\nabla F(y)|_{\mathbb{C}}^2\:\nu(dx)=\int_K\left\|(f\otimes g)_x\right\|_{\mathcal{H},x}^2\nu(dx).\]
Using bilinearity and the denseness of $\mathcal{S}_\mathbb{C}$, this extends to an isometric embedding 
\[\kappa : \mathcal{H}_\mathbb{C} \to \int_K^\otimes \mathcal{H}_{\mathbb{C},x}\nu(dx)\]
of $\mathcal{H}_\mathbb{C}$ into $\int_K^\otimes \mathcal{H}_{\mathbb{C},x}\nu(dx)$. In fact $\kappa$ is onto: Assume that $\omega\in \int_K^\oplus \mathcal{H}_{\mathbb{C},x}\nu(dx)$ is such that 
\[0=\left\langle\omega, f\otimes g\right\rangle_\mathcal{H}=\int_K \left\langle\omega_x,(f\otimes g)_x\right\rangle_{\mathcal{H},x}\nu(dx)=\int_K\overline{g(x)}\left\langle \omega_x,(f\otimes\mathbf{1})_x\right\rangle_{\mathcal{H},x}\nu(dx)\]
for all $f\otimes g\in\mathcal{S}_\mathbb{C}$. Note that we have used (\ref{E:pullout}). Then necessarily $\left\langle \omega_x,(f\otimes\mathbf{1})_x\right\rangle_{\mathcal{H},x}=0$ for $\nu$-a.a. $x$. Therefore $\omega_x$ must be zero on $\mathcal{H}_{\mathbb{C},x}$ for such $x$ and integrating, we have $\omega=0$ in $\int_K^\otimes \mathcal{H}_{\mathbb{C},x}\nu(dx)$.
\end{proof}

The definitions 
\begin{equation}\label{E:actions}
c(a\otimes b):=(ca)\otimes b - c(bd) \ \ \text{ and } \ \ (a\otimes b)d:=a\otimes (bd)
\end{equation}
for $a,b,c\in\mathcal{F}_\mathbb{C}\otimes\mathcal{F}_\mathbb{C}$ and $d\in L_{\infty,\mathbb{C}}(K,\nu)$ extend continuously to uniformly bounded actions on $\mathcal{H}_\mathbb{C}$, 
\begin{equation}\label{E:boundedactions}
\left\| c\omega\right\|_\mathcal{H}\leq \left\| c\right\|_{L_\infty(X,m)}\left\|\omega\right\|_\mathcal{H} \ \ \text{ and } \ \
 \left\| \omega c\right\|_\mathcal{H}\leq \left\| c\right\|_{L_\infty(X,m)}\left\|\omega\right\|_\mathcal{H}, \ \ \omega\in\mathcal{H}_\mathbb{C}.
\end{equation}
For $\omega\in\mathcal{H}_\mathbb{C}$ and $c\in\mathcal{F}_\mathbb{C}$ we have the equality $c\omega=\omega c$ in $\mathcal{H}_\mathbb{C}$. See \cite{CS03, CS09, IRT}. By
\[\nu_{\omega,\eta}(A):=\left\langle \omega\mathbf{1}_A, \eta\right\rangle_\mathcal{H}\]
for $\omega,\eta\in\mathcal{H}_\mathbb{C}$ and Borel set $A\subset K$, we define \emph{(weighted) energy measures} for $1$-forms. Note that for any $f\in\mathcal{F}_\mathbb{C}$ we have $\nu_{f\otimes\mathbf{1}}=\nu_f$. Note also that every $\nu_{\omega,\eta}$ is absolutely continuous with respect to $\nu$, and $x\mapsto \left\langle \omega_x,\eta_x\right\rangle_{\mathcal{H},x}$ is a version of its density. To stress the similarity to the energy density we also use the notation
\begin{equation}\label{E:densitynotation}
\Gamma_{\mathcal{H},x}(\omega_x,\eta_x):=\left\langle \omega_x,\eta_x\right\rangle_{\mathcal{H},x}.
\end{equation}

A derivation operator $\partial: \mathcal{F}_\mathbb{C}\to \mathcal{H}_\mathbb{C}$ can be defined by setting 
\[\partial f:= f\otimes \mathbf{1}.\] 
It satisfies the Leibniz rule,
\begin{equation}\label{E:Leibniz}
\partial(fg)=f\partial g + g\partial f, \ \ f,g \in \mathcal{F}_\mathbb{C}.
\end{equation}
The linear operator $\partial$ is bounded, more precisely,
\begin{equation}\label{E:normandenergy}
\left\|\partial f\right\|_\mathcal{H}^2=\mathcal{E}(f), \ \ f\in\mathcal{F}_\mathbb{C}.
\end{equation}

By the closedness of $(\mathcal{E},\mathcal{F}_\mathbb{C})$ in $L_{2,\mathbb{C}}(K,\nu)$ the derivation $\partial$ may be viewed as an unbounded closed operator from $L_2(K,\nu)$ into $\mathcal{H}$. In coordinates $y=y(x)$, the operator $\partial$ agrees with the usual gradient operator $\nabla$ in $\mathbb{R}^2$. This can be phrased as a Corollary of Theorem \ref{T:fibers} in terms of the isomorphisms $\Phi_x$ from (\ref{E:fiberiso}). It may be viewed as a \emph{'pointwise formula'} for the derivation $\partial$.

\begin{corollary}\label{C:gradient}
For any $f=F\circ h$ and $g=G\circ h$ with $F\in C^1_\mathbb{C}(\mathbb{R}^2)$ and $G$ bounded Borel measurable on $\mathbb{R}^2$ we have
\[\Phi_x((g\partial f)_x)=Z_xG(y)\nabla F(y)\]
for $\nu$-a.e. $x\in K$, and in particular, 
\[\Phi_x((\partial f)_x)=Z_x\nabla F(y).\]
\end{corollary}

Because of the self-duality of $\mathcal{H}_\mathbb{C}$ we regard its elements also as \emph{vector fields} and $\partial$ as a generalization of the classical \emph{gradient} operator. Let $\mathcal{F}^\ast_\mathbb{C}$ denote the dual of $\mathcal{F}_\mathbb{C}/\sim$ with the norm
\[\left\|u\right\|_{\mathcal{F}^\ast}=\sup\left\lbrace |u(f)|: f\in\mathcal{F}_\mathbb{C}, \ \mathcal{E}(f)\leq 1\right\rbrace.\] 
The space $\mathcal{F}^\ast_\mathbb{C}$ may be thought of as a \emph{'space of distributions'}. The symbol $\left\langle \cdot,\cdot\right\rangle$ will denote the dual pairing  between $\mathcal{F}_\mathbb{C}/\sim$ and $\mathcal{F}^\ast_\mathbb{C}$. Note that $\mathcal{F}_\mathbb{C}\subset \mathcal{F}^\ast_\mathbb{C}$ and $\left\langle f, g\right\rangle=\left\langle f,g\right\rangle_{L_2(K,\nu)}$ for $f,g\in \mathcal{F}_\mathbb{C}/\sim$.

Given a vector field of form $g\partial f$ with $f,g\in\mathcal{B}_\mathbb{C}$, its \emph{divergence} $\partial^\ast(g\partial f)$ can be defined similarly as in \cite{HRT} by
\[\partial^\ast(g\partial f)(\varphi):=-\int_Kg\:\Gamma(f,\overline{\varphi})\:d\nu,\]
seen as an element in $\mathcal{F}^\ast_\mathbb{C}$. This extends continuously to a bounded linear operator $\partial^\ast$ from $\mathcal{H}_\mathbb{C}$ into $\mathcal{F}^\ast_\mathbb{C}$ such that 
\[\partial^\ast v(\varphi)=-\left\langle v,\partial \overline{\varphi}\right\rangle_\mathcal{H},\ \ v\in\mathcal{H}_\mathbb{C},\ \varphi\in\mathcal{F}_\mathbb{C},\]
see \cite[Lemma 3.1]{HRT} for a proof. Note that here the divergence $\partial^\ast v$ is defined in a \emph{distributional sense}. By restricting the space of test functions further this definition can be  modified to fit into the context of distributions on p.c.f. fractals as studied in \cite{RoStr10}. In coordinates we have the following expression.

\begin{theorem}\label{T:divergence}
Let $f=F\circ h$, $g=G\circ h$ and $u=U\circ h$ with $F,U\in C^1_\mathbb{C}(\mathbb{R}^2)$ and $G$ bounded Borel measurable on $\mathbb{R}^2$. Then
\[\partial^\ast(g\partial f)(u)=\int_K \diverg_{Z_x}(G\nabla F)(U)(y) d\nu,\]
where for $\nu$-a.e. $x\in K$,
\begin{equation}\label{E:divcoords}
\diverg_{Z_x}(G\nabla F)(\cdot)(y):=-\sum_{ij}Z_{ij}(x)(G\nabla F)_j(y)\frac{\partial (\cdot)}{\partial y_i}(y).
\end{equation}
\end{theorem}

(\ref{E:divcoords}) may be seen as a bounded linear functional on $C^1_\mathbb{C}(\mathbb{R}^2)$. In a sense it remotely reminds of the divergence in Riemannian coordinates.

\begin{proof}
We have 
\begin{align}
-\left\langle g\partial f,\partial \overline{u}\right\rangle_\mathcal{H}&=-\int_K\left\langle G(y)\nabla F(y),Z_x\nabla \overline{U}(y)\right\rangle_{\mathbb{C}^2}\nu(dx)\notag\\
&=-\int_K\sum_{ij}Z_{ij}(x)G(y)\frac{\partial F(y)}{\partial y_j}\frac{\partial U}{\partial y_i}(y)\nu(dx).\notag
\end{align}
\end{proof}

Let $\Delta_\nu$ denote the \emph{energy Laplacian} on $K$, that is, the infinitesimal generator of $(\mathcal{E},\mathcal{F})$ in $L_2(K,\nu)$. Its complexification is the non-positive definite self-adjoint operator $\Delta_\nu$ uniquely associated with the closed form $(\mathcal{E},\mathcal{F}_\mathbb{C})$ on $L_{2,\mathbb{C}}(K,\nu)$ by $\mathcal{E}(f,g)=-\left\langle f,\Delta_\nu g\right\rangle$ for $f,g\in \mathcal{F}_\mathbb{C}$. Since $-\partial^\ast\partial g(f)=\mathcal{E}(f,g)$ we observe 
\begin{equation}\label{E:generatoridentity}
\Delta_\nu g=\partial^\ast\partial g
\end{equation}
for any $g\in \mathcal{F}_\mathbb{C}$. Here (\ref{E:generatoridentity}) is seen as an equality in $\mathcal{F}_\mathbb{C}^\ast$, below we will refer to an $L_2$-context. In coordinates we have the following.

\begin{theorem}\label{T:Laplace}
For any $f=F\circ h$ with $F\in C^2_\mathbb{C}(\mathbb{R}^2)$ we have $f\in dom\:\Delta_\nu$ and
\[\Delta_\nu f(x)=\tr(Z_x \Hess F)(y)\] 
for $\nu$-a.e. $x\in K$. Moreover, given arbitrary $u=U\circ h$ with $U\in C^1_\mathbb{C}(\mathbb{R}^2)$, the identity
\[\int_K \diverg_{Z_x}(\nabla F)(\overline{U})(y)d\nu=\int_K\Delta_\nu f\:\overline{u}\:d\nu\]
holds.
\end{theorem}

\begin{proof}
The first statement is a special case of \cite[Theorem 8]{T08}. For the second note that
\begin{align}
\int_K\diverg_{Z_x}(\nabla F)(\overline{U})(y)\nu(dx)&=-\int_K\left\langle \nabla F(y), Z_x\nabla U(y)\right\rangle_{\mathbb{C}^2}\nu(dx)\notag\\
&=-\mathcal{E}(f,u)\notag\\
&=\int_K\Delta_\nu f \overline{u}\:d\nu.\notag 
\end{align}
\end{proof}

\begin{remark}
In the classical Euclidean case we have
\[\diverg (G\nabla F)=G\Delta F+\nabla F\nabla G\]
for $G\in C^1(\mathbb{R}^2)$ and $F\in C^2(\mathbb{R}^2)$, where $\nabla$, $\Delta$ and $\diverg$ denote the Euclidean gradient, Laplacian and divergence, respectively. In our case we have 
\[\partial^\ast(g\partial f)=g\Delta_\nu f+\Gamma(f,g)\]
for $f,g\in\mathcal{F}$, by \cite[Lemma 3.2]{HRT}, seen as an equality in $\mathcal{F}_\mathbb{C}^\ast$. This may again be written in coordinates. Given $f=F\circ h$ and $g=G\circ h$ with $F\in C^2(\mathbb{R}^2)$ and $G\in C^2(\mathbb{R}^2)$ we have
\begin{multline}
\int_K \diverg_{Z_x}(G\nabla F)(U)(y)d\nu\notag\\
=\int_KG(y)\tr(Z_x\Hess F)(y)U(y)d\nu+\int_K\left\langle \nabla F(y),Z_x\nabla G(y)\right\rangle_{\mathbb{C}^2}U(y)d\nu
\end{multline}
for any $U\in C^1(\mathbb{R}^2)$, as may be seen from the previous theorems.
\end{remark}

The divergence $\partial^\ast$ may also be seen in an \emph{$L_2$-sense} as an unbounded linear operator from $\mathcal{H}_\mathbb{C}$ into $L_{2,\mathbb{C}}(K,\nu)$. As usual an element $v\in\mathcal{H}_\mathbb{C}$ is said to be a member of $dom\:\partial^\ast$ if
there is some $v^\ast\in L_{2,\mathbb{C}}(K,\nu)$ such that $\left\langle v^\ast, u\right\rangle_{L_2(K,\nu)}=-\left\langle v, \partial u\right\rangle_\mathcal{H}$ for all $u\in\mathcal{F}_\mathbb{C}$. In this case we set $\partial^\ast v:=v^\ast$. The operator $-\partial^\ast$ is the adjoint (codifferential) of $\partial$, and since $\partial$ is a closed operator, $dom\:\partial^\ast$ is dense in $\mathcal{H}_\mathbb{C}$. Note that for the previous distributional definition we obtain 
\[\partial^\ast v(\varphi)=\left\langle \partial^\ast v,\overline{\varphi}\right\rangle_{L_2(K,\nu)}.\] 
We end this section by a remark that allows to retrieve some more explicit information about the domain $dom\:\partial^\ast$ of the divergence operator $\partial^\ast$ in $L_2$-sense. 

\begin{remark}\label{R:ortho}
As the quadratic form $(\mathcal{E},\mathcal{F}_\mathcal{C})$ is closed, the operator $\partial$ is closed. From (\ref{E:resistest}) it follows that $(\mathcal{E},\mathcal{F}_\mathcal{C})$ has a spectral gap. Therefore the range $Im\:\partial$ of $\partial$ is a closed subspace of $\mathcal{H}_\mathbb{C}$ and the space $\mathcal{H}_\mathbb{C}$ decomposes orthogonally into $Im\:\partial$ and its complement $(Im\:\partial)^\bot$. By orthogonality we observe $(Im\:\partial)^\bot=ker\:\partial^\ast$.
In \cite{HRT} and \cite{IRT} the space $ker\:\partial^\ast$ had been identified as the \emph{space of harmonic forms (or vector fields)} in the sense of Hodge theory if the topological dimension is one.
\end{remark}

Let $dom\:\Delta_\nu$ denote the domain of $\Delta_\nu$ in $L_{2,\mathbb{C}}(K,\nu)$. We obtain the following explicit description of $dom\:\partial^\ast$.
\begin{corollary}
We have
\[dom\:\partial^\ast=\left\lbrace v\in\mathcal{H}_\mathbb{C}: v=\partial f+w: f\in dom\:\Delta_\nu\ ,\ w\in ker\:\partial^\ast\right\rbrace.\] 
For any $v=\partial f+w$ with $f\in dom\:\Delta_\nu$ and $w\in ker\:\partial^\ast$ we have
\[\partial^ \ast v=\Delta_\nu f,\]
and for any $g\in dom\:\Delta_\nu$ formula (\ref{E:generatoridentity}) holds in $L_2(K,\nu)$.
\end{corollary}

Note that the right hand side of this identity is explicitely seen to be dense in $\mathcal{H}$: The denseness of the range of Green's operator in $L_2(K,\nu)$ can be used to see that $dom\:\Delta_\nu$ is dense in $\mathcal{F}_\mathbb{C}$. Therefore the image $\partial(dom\:\Delta_\nu)$ of $dom\:\Delta_\nu$ under the derivation $\partial$ is dense in $Im\:\partial$.

\begin{proof}
According to Remark \ref{R:ortho} any $v\in dom\:\partial^\ast\subset\mathcal{H}$ admits a unique orthogonal decomposition $v=\partial f+w$ with $f\in\mathcal{F}_\mathbb{C}$ and $w\in ker\:\partial^\ast$. Since
\[\Delta_\nu f=\partial^\ast\partial f=\partial^\ast v\]
is in $L_2(K,\nu)$ we obtain $f\in dom\:\Delta_\nu$. Conversely, any vector field $v=\partial f+w\in\mathcal{H}$ with $f\in dom\:\Delta_\nu$ and $w\in ker\:\partial^\ast$ is a member of $dom\:\partial^\ast$ since for any $u\in \mathcal{F}_\mathbb{C}$ we have
\[-\left\langle v,\partial u\right\rangle_\mathcal{H}=\left\langle \partial f,\partial u\right\rangle_\mathcal{H}=\mathcal{E}(f,u)=-\left\langle \Delta_\nu f, u\right\rangle_{L_2(K,\nu)}.\] 
The last statements of the Corollary are obvious consequence.
\end{proof}

\begin{remark}\mbox{}
\begin{enumerate}
\item[(i)] Assume $f=F\circ h$. If $F\in C_\mathbb{C}^2(\mathbb{R}^2)$ then $\partial f\in dom\:\partial^\ast$ by Theorem \ref{T:Laplace}. For general $F\in C_\mathbb{C}^1(\mathbb{R}^2)$, however, we will not have $f\in dom\:\Delta_\nu$ and therefore also not $\partial f\in dom\:\partial^ \ast$.
\item[(ii)] Even if $f=F\circ h$ with $F\in C_\mathbb{C}^2(\mathbb{R}^2)$ a simple vector field $g\partial f$ can generally not be expected to be an element of $dom\:\partial^\ast$. Let $\mathcal{P}_{Im\:\partial}$ denote the orthogonal projection in $\mathcal{H}_\mathbb{C}$ onto $Im\:\partial$. Given some measurable (or energy finite) $g$, the operator $f\mapsto \mathcal{P}_{Im\:\partial}(g\partial f)$ is a complicated non-local first order pseudo-differential operator. See \cite{IRS} for some first results concerning pseudo-differential operators on fractals.
\end{enumerate}
\end{remark}

\section{Dirac operators}\label{S:Dirac}

The definitions of the derivation $\partial$ and the divergence $\partial^\ast$ allow to define related Dirac operators. Consider the Hilbert space
\[L_{2,\mathbb{C}}(K,\nu)\oplus \mathcal{H}_\mathbb{C}\]
and write 
\[\left\langle (f,\omega), (g,\eta)\right\rangle_\oplus:=\left\langle f,g\right\rangle_{L_2(K,\nu)}+\left\langle \omega,\eta\right\rangle_{\mathcal{H}}\]
for its scalar product. We define the \emph{Dirac operator} associated with $(\mathcal{E},\mathcal{F})$ to be the unbounded linear operator
$D$ on $L_{2,\mathbb{C}}(K,\nu)\oplus \mathcal{H}_\mathbb{C}$ with domain
\[dom\:D:=\mathcal{F}_\mathbb{C}\oplus dom\:\partial^\ast,\]
given by 
\[D(f,\omega):=(-i\partial^\ast\omega, -i\partial f), \ \ (f,\omega)\in dom\:D.\]
In matrix notation we have
\[D=\left(\begin{array}{rr}
0 & -i\partial^\ast\\
-i\partial & 0\end{array}\right).\]

\begin{remark}
Here the signs and imaginary factors are chosen in a way that fits with the following section. Of course other, possibly more physical choices, can be considered by obvious modifications.  
\end{remark}

The next theorem was obtained in  abstract form in \cite{CS03}, but for the convenience of the reader we sketch a proof. Note that the results in Section~\ref{S:vector} show that 
$D$ is a pointwise defined local operator. 

\begin{theorem}\label{T:Dirac}
The operator $(D, dom\:D)$ is self-adjoint operator on $L_{2,\mathbb{C}}(K,\nu)\oplus \mathcal{H}_\mathbb{C}$.
\end{theorem}

\begin{proof}
Obviously $dom\:D$ is dense in $L_{2,\mathbb{C}}(K,\nu)\oplus \mathcal{H}_\mathbb{C}$. Recall that by definition
\begin{multline}
dom\:D^\ast=\left\lbrace (g,\eta)\in L_{2,\mathbb{C}}(X,m)\oplus \mathcal{H}_\mathbb{C}: \text{ there exists some $(g,\eta)^\ast \in L_{2,\mathbb{C}}(X,m)\oplus \mathcal{H}_\mathbb{C}$}\right.\notag\\
\left. \text{ such that $\left\langle (f,\omega),(g,\eta)^\ast \right\rangle_{\oplus}=\left\langle D(f,\omega),(g,\eta)\right\rangle_\oplus$ for all $(f,\omega)\in dom\:D$} \right\rbrace.
\end{multline}
Note that $(-i\partial)^\ast=-i\partial^\ast$. For arbitrary $(f,\omega)$ and $(g,\eta)$ from $dom\:D$ we have
\begin{align}
\left\langle (-i\partial^\ast\omega,-i\partial f), (g,\eta)\right\rangle_\oplus &=\left\langle-i\partial^\ast\omega,g\right\rangle_{L_2(X,m)}+\left\langle -i\partial f,\eta\right\rangle_\mathcal{H}\notag\\
&=\left\langle \omega, -i\partial g\right\rangle_\mathcal{H}+\left\langle f,-i\partial^\ast\eta\right\rangle_{L_2(X,m)}\notag\\
&=\left\langle (\omega,f), (-i\partial^\ast\eta,-i\partial g)\right\rangle_\oplus.\notag
\end{align}
Consequently $dom\:D\subset dom\:D^\ast$ and $D^\ast\omega=D\omega$ for all $\omega\in dom\:D$.
\end{proof}

Our next aim is to justify our nomenclature by showing that in an appropriate sense, $D^2:=D\circ D$ is the Laplacian. The Kusuoka Laplacian $\Delta_\nu$ acts on functions, and we need to discuss a second Laplacian which acts on $1$-forms. From \cite{HT} we recall the following.

Let
\[dom\:\Delta_{\nu,1}:=\left\lbrace \omega\in dom\:\partial^\ast: \partial^\ast\omega\in\mathcal{F}_\mathbb{C}\right\rbrace\]
and for $\omega\in dom\:\Delta_{\nu,1}$, define
\begin{equation}\label{E:formlaplace}
\Delta_{\nu,1}\omega:=\partial\partial^\ast\omega.
\end{equation}
To $\Delta_{\nu,1}$ we refer as the \emph{form Laplacian} associated with $(\mathcal{E},\mathcal{F})$. The following result had been shown in \cite[Theorem 6.1]{HT}. For convenience, we briefly sketch its proof.

\begin{theorem}\label{T:formlaplace}
Definition (\ref{E:formlaplace}) yields a self-adjoint operator $(\Delta_{\nu,1}, dom\:\Delta_{\nu,1})$ on  $\mathcal{H}_\mathcal{C}$.
\end{theorem}

\begin{proof}
By Remark \ref{R:ortho} we observe $(Im\:\partial)^\bot\subset dom\:\Delta_{\nu,1}$. Let $\mathcal{C}$ be an $\mathcal{E}$-dense subspace of $\mathcal{F}_\mathbb{C}$ such that for all $f\in\mathcal{C}$, we have $\Delta_\nu f\in\mathcal{F}_\mathbb{C}$. Such a space always exists, for instance, we may use the image of $\mathcal{F}_\mathbb{C}$ under the Green operator $\Delta_\nu^{-1}$. For $\partial f$ with $f\in \mathcal{C}$ and $\Delta_\nu f=g\in\mathcal{F}_\mathbb{C}$ we have
\[\Delta_{\nu,1}(\partial f)=\partial(\Delta_\nu f)=\partial g.\]
Therefore also $\partial(\mathcal{C})\subset dom\:\Delta_{\nu,1}$. As $\mathcal{C}$ is $\mathcal{E}$-dense in $\mathcal{F}_\mathbb{C}$, its image $\partial(\mathcal{C})$ is dense in $Im\:\partial$, and therefore $dom\:\Delta_{\nu,1}$ is dense in $\mathcal{H}_\mathbb{C}$. For any $\omega\in dom\:\Delta_{\nu,1}$ the identity
\[\left\langle \eta,\Delta_{\nu,1}\omega\right\rangle_\mathcal{H}=-\left\langle \partial \eta, \partial^\ast\omega\right\rangle_{L_2(K,\nu)}=\left\langle \Delta_{\nu,1}\eta,\omega\right\rangle_\mathcal{H}, \ \ \eta\in dom\:\Delta_{\nu,1},\]
showing $dom\:\Delta_{\nu,1}\subset dom\:\Delta_{\nu,1}^\ast$ and $\Delta_{\nu,1}^\ast \omega=\Delta_{\nu,1}\omega$ for all $\omega\in dom\:\Delta_{\nu,1}$.
\end{proof}

Now set $dom\:D^2:=dom\:\Delta_\nu\oplus dom\:\Delta_{\nu,1}$, clearly a dense subspace of $L_{2,\mathbb{C}}(K,\nu)\oplus \mathcal{H}_\mathbb{C}$.
Then the following is immediate.
\begin{lemma}
For any $(f,\omega)\in dom\:\Delta_\nu\oplus dom\:\Delta_{\nu,1}$ we have $D^2(f,\omega)=(\Delta_\nu f, \Delta_{\nu,1}\omega)$.
\end{lemma}
In matrix notation this is
\[D^2=\left(\begin{array}{ll}
\Delta_\nu & 0\\
0 &\Delta_{\nu,1}\end{array}\right).\]

\section{Magnetic Schr\"odinger operators}\label{S:magnetic}

The notions defined in Section \ref{S:vector} also allow to define a suitable generalization of the quantum mechanical Hamiltonian
\begin{equation}\label{E:Hamiltonian}
H^{A,V}=(-i\nabla-A)^2+V
\end{equation}
for a particle moving in Euclidean space $\mathbb{R}^n$ subject to a real (\emph{magnetic}) vector potential $A$ and a real valued (\emph{electric}) potential $V$. The main result of this section is Theorem \ref{T:main} below, which tells that there exists a self-adjoint operator on $L_{2,\mathbb{C}}(K,\nu)$ which generalizes (\ref{E:Hamiltonian}). Another result, Theorem \ref{T:gauge}, is a gauge invariance result and tells that, roughly speaking, we may restrict attention to divergence free vector potentials.

We collect some preliminaries. Let 
\[\mathcal{H}=\left\lbrace v\in\mathcal{H}_\mathbb{C}: \overline{v}=v\right\rbrace\]
be the \emph{space of real vector fields}. An element $a\in\mathcal{H}$ may be seen as a bounded linear mapping $a:\mathcal{F}_\mathbb{C}\to\mathcal{H}_\mathbb{C}$ defined by $f\mapsto fa$.
I admits the estimate $\left\| fa\right\|_{\mathcal{H}}\leq \left\|f\right\|_{L_\infty(K,\nu)}\left\|a\right\|_{\mathcal{H}}$. For $a\in\mathcal{H}$, define a bounded linear mapping $a^\ast:\mathcal{H}_\mathbb{C} \to \mathcal{F}^\ast_\mathbb{C}$ by $v\mapsto a^\ast v$,
where
\[(a^\ast v)(\varphi):=\int_K\varphi \Gamma_{\mathcal{H}}(a,\overline{v})\:d\nu=\left\langle a\varphi,\overline{v}\right\rangle_\mathcal{H},\ \ \varphi\in \mathcal{F}_\mathbb{C}.\]
The boundedness follows from
\[\left\|a^\ast v\right\|_{\mathcal{F}^\ast}=\sup\left\lbrace |\int_K\varphi \Gamma_\mathcal{H}(a,\overline{v})\:d\nu|: \varphi\in\mathcal{F}_\mathbb{C}/\sim \text{ with $\mathcal{E}(\varphi)\leq 1$} \right\rbrace\]
together with 
\[|\int_K\varphi \Gamma_\mathcal{H}(a,\overline{v})\:d\nu|\leq \left\|\varphi\right\|_{L_\infty(K,\nu)}|\left\langle a,\overline{v}\right\rangle_\mathcal{H}|\leq\mathcal{E}(\varphi)^{1/2}\left\|a\right\|_\mathcal{H}\left\|v\right\|_\mathcal{H},\]
where we have used (\ref{E:resistest}) and the Cauchy-Schwarz inequality in $\mathcal{H}_\mathbb{C}$. In terms of duality we observe $a^\ast v(\varphi)=\left\langle a^\ast v,\overline{\varphi}\right\rangle$. Given $f\in \mathcal{F}_\mathbb{C}$ we have in particular $a^\ast af:=a^\ast (af)\in\mathcal{F}_\mathbb{C}^\ast$ with
\[(a^\ast af)(\varphi)=\int_K\varphi f\Gamma_{\mathcal{H}}(a)\:d\nu.\]

Now let $a\in \mathcal{H}$ and $V\in L_\infty(K,\nu)$. The preceding considerations ensure that 
\[\mathcal{E}^{a,V}(f,g):=\left\langle (-\Delta_\nu+2ia^\ast\partial+a^\ast a+i(\partial^\ast a)+V)f,g\right\rangle,\]
$f,g\in \mathcal{F}_{\mathbb{C}}$, provides a well-defined form $\mathcal{E}^{a,V}$ on $\mathcal{F}_\mathbb{C}\times \mathcal{F}_\mathbb{C}$
that is linear in the first and conjugate linear in the second argument.

\begin{proposition}\label{P:quadratic} 
Let $a\in \mathcal{H}$ and $V\in L_\infty(K,\nu)$.
\begin{enumerate}
\item[(i)] For any $f,g\in \mathcal{F}_{\mathbb{C}}$ we have
\begin{equation}\label{E:newhamiltonian}
\mathcal{E}^{a,V}(f,g)=\left\langle (-i\partial-a)f,(-i\partial-a)g\right\rangle_{\mathcal{H}}+\left\langle fV,g\right\rangle.
\end{equation}
\item[(ii)] The form $\mathcal{E}^{a,V}$ is a quadratic form on $\mathcal{F}_\mathbb{C}$.
\end{enumerate}
\end{proposition}

Before proving Proposition \ref{P:quadratic}, we verify a product rule. For $u\in\mathcal{F}_\mathbb{C}^\ast$ and $f\in\mathcal{F}_\mathbb{C}$  define the \emph{product} $fu\in\mathcal{F}_\mathbb{C}^\ast$ by
\[(fu)(\varphi):=\left\langle u,\overline{\varphi f}\right\rangle,\ \ \varphi\in\mathcal{F}_\mathbb{C}.\]
In particular, we ob obtain
\[(f(\partial^\ast a))(\varphi)=\left\langle \partial^\ast a, \overline{\varphi f}\right\rangle=-\left\langle a,\partial(\overline{\varphi f})\right\rangle_\mathcal{H}.\]

\begin{lemma}\label{L:Leibniz}
For $a\in \mathcal{H}$ and $f\in \mathcal{F}_\mathbb{C}$ we have 
\[\partial^\ast(fa)=f(\partial^\ast a)+a^\ast \partial f,\]
seen as an equality in $\mathcal{F}_\mathbb{C}^\ast$.
\end{lemma}

\begin{proof} Using the Leibniz rule (\ref{E:Leibniz}) we see that 
\[(f(\partial^\ast a))(\varphi)=-\left\langle a,\partial(\overline{f\varphi})\right\rangle_\mathcal{H}=-\left\langle a,\overline{\varphi}\partial \overline{f}\right\rangle_\mathcal{H}-\left\langle a, \overline{f}\partial\overline{\varphi}\right\rangle_\mathcal{H}=-(a^\ast\partial f)(\varphi)+\partial^\ast(fa)(\varphi)\]
for any $\varphi\in\mathcal{F}_\mathbb{C}$.
\end{proof}

We prove Proposition \ref{P:quadratic}.
\begin{proof}
The right hand side of (\ref{E:newhamiltonian}) rewrites 
\begin{equation}\label{E:RHSofE}
\mathcal{E}(f,g)+M(f,g)+\left\langle fV,g\right\rangle ,
\end{equation}
where
\[M(f,g):=i\left\langle \partial f, ag\right\rangle_\mathcal{H}-i\left\langle af,\partial g\right\rangle_\mathcal{H}+\left\langle fa,ga\right\rangle_\mathcal{H}.\]
For the first summand on the right hand side we observe
\[\left\langle i\partial f, ag\right\rangle_\mathcal{H}=\left\langle\partial f,-iag\right\rangle_\mathcal{H}=\overline{\left\langle -iag,\partial f\right\rangle_\mathcal{H}}\]
and similarly for the second. Consequently $M(g,f)=\overline{M(f,g)}$, and and therefore the expression (\ref{E:RHSofE}) is conjugate symmetric. We show that (\ref{E:RHSofE}) equals $\mathcal{E}^ {a,V}(f,g)$. Note first that
\[\left\langle af,\partial g\right\rangle_\mathcal{H}=\int_Xfd\Gamma_\mathcal{H}(a,\partial g)=a^\ast\partial\overline{g}(f).\]
By Lemma \ref{L:Leibniz} this equals 
\[\partial^\ast(\overline{g}a)(f)-\overline{g}(\partial^\ast a(f)=-\left\langle a\overline{g},\partial\overline{f}\right\rangle_\mathcal{H}-\left\langle\partial^\ast a,g\overline{f}\right\rangle_\mathcal{H}=-\left\langle \partial f, ag\right\rangle_\mathcal{H}-\left\langle (\partial^\ast a)f,g\right\rangle_\mathcal{H}.\]
Next, note that
\[\left\langle \partial f, ag\right\rangle_\mathcal{H}=\int_K\overline{g}\Gamma_\mathcal{H}(\partial f, a)d\nu=a^\ast(\partial f)(\overline{g})=\left\langle a^\ast(\partial f),g\right\rangle.\]
Consequently
\[M(f,g)=2i\left\langle a^\ast \partial f,g\right\rangle+i\left\langle(\partial^\ast a)f,g\right\rangle+\left\langle a^\ast af,g\right\rangle,\]
and taking into account (\ref{E:generatoridentity}), identity (\ref{E:newhamiltonian}) in (i) follows. Statement (ii) is a straightforward consequence.
\end{proof}

Recall (\ref{E:densitynotation}). To
\begin{equation}\label{E:Hinfty}
\mathcal{H}_\infty:=\left\lbrace \omega\in\mathcal{H}: \overline{v}=v \text{ and } \Gamma_{\mathcal{H},\cdot}(v) \in L_\infty(K,\nu)\right\rbrace 
\end{equation}
we refer as the space of \emph{(real) vector fields of bounded length}. Assume $a\in\mathcal{H}_\infty$. Then the multiplication $f\mapsto fa$ may be seen as bounded linear operator $a:L_{2,\mathbb{C}}(K,\nu)\to\mathcal{H}_\mathbb{C}$ by $f\mapsto fa$,
because
\begin{equation}\label{E:multbyabetter}
\left\| fa\right\|_{\mathcal{H}}^2=\int_K |f|^2\Gamma_{\mathcal{H},\cdot}(a)d\nu\leq \left\|f\right\|_{L_2(K,\nu)}^2\left\|\Gamma_{\mathcal{H},\cdot}(a)\right\|_{L_\infty(K,\nu)}.
\end{equation}
From (\ref{E:normdef}) it easily follows that
\[\mathcal{S}=\lin\left\lbrace f\otimes g: f=F\circ h, g=G\circ h\ \text{ with $F,G\in C^1(\mathbb{R}^2)$}\right\rbrace.\]
is a subset of $\mathcal{H}_\infty$.

The estimate (\ref{E:multbyabetter}) allows the application of classical perturbation arguments to prove the closedness of $\mathcal{E}^{a,V}$ and the essential self-adjointness of an analog of (\ref{E:Hamiltonian}).

\begin{theorem}\label{T:main}
Let $a\in \mathcal{H}_\infty$ and $V\in L_\infty(K,\nu)$. 
\begin{enumerate}
\item[(i)] The bilinear form $(\mathcal{E}^{a,V},\mathcal{F}_\mathbb{C})$ is closed. 
\item[(ii)] The unique self-adjoint non-negative definite operator associated with $(\mathcal{E}^{a,V},\mathcal{F}_\mathbb{C})$ 
is given by 
\[H^{a,V}=(-i\partial-a)^\ast(-i\partial-a)+V,\]
and the domain of $\Delta_\nu$ is a domain of essential self-adjointness for $H^{a,V}$. 
\end{enumerate}
\end{theorem}

\begin{proof}
Recall that $\mathcal{E}^{a,V}(f,g)$ equals (\ref{E:RHSofE}) for any $f,g\in\mathcal{F}_\mathbb{C}$. We show that for any $\varepsilon>0$ and any $f\in \mathcal{F}_\mathbb{C}$, 
\begin{equation}\label{E:KLMN}
|M(f,f)|\leq \varepsilon^2 \mathcal{E}(f)+C_{a,V}\left\|f\right\|_{L_2(K,\nu)}^2,
\end{equation}
where $C_{a,V}>0$ is a constant bounded by $\left\|V\right\|_{L_\infty(K,\nu)}$ plus a multiple of $\left\|\Gamma_{\mathcal{H},\cdot}(a)\right\|_{L_\infty(K,\nu)}$. Then the result follows by the classical KLMN theorem, cf. \cite[Theorem X.17]{RS}. By the boundedness of $V$ clearly
\[|\left\langle fV,f\right\rangle_{L_2(K,\nu)}|\leq \left\|V\right\|_{L_\infty(K,\nu)}\left\|f\right\|_{L_2(K,\nu)}^2.\]
The bound (\ref{E:multbyabetter}) covers the summand $\left\langle fa,fa\right\rangle_\mathcal{H}$, and applying it once more,
\[|\left\langle \partial f, fa\right\rangle_\mathcal{H}|\leq \left\|\partial f\right\|_\mathcal{H}\left\|fa\right\|_\mathcal{H}\leq
\frac{1}{2}(\varepsilon^2\left\|\partial f\right\|_\mathcal{H}^2+\frac{1}{\varepsilon^2}\left\|\Gamma_{\mathcal{H},\cdot}(a)\right\|_{L_\infty(K,\nu)}\left\|f\right\|_{L_2(K,\nu)}^2).\]
Taking into account also (\ref{E:normandenergy}), we arrive at (\ref{E:KLMN}).
\end{proof}

\begin{remark}
Theorem  \ref{T:main} is a rather simple result and can certainly be improved. For instance, to obtain essential self-adjointness for magnetic Schr\"odinger operators on Euclidean spaces it is sufficient that the vector potential $A$ is locally in $L_4$ and $\diverg A$ is locally in $L_2$, while the scalar potential $V$ may be taken to be locally in $L_2$. See \cite{L83}. What we would like to point out here is that any analog of such a hypothesis will again require the weighted energy measure $\nu_a$ of $a$ to be absolutely continuous with respect to the reference measure $\nu$.
\end{remark}

If in (\ref{E:Hamiltonian}) a gradient $\nabla\Lambda$ of a real-valued function $\Lambda$ is added to the vector potential $A$, then 
the operator $H^{A+\nabla\Lambda,V}$ is unitarily equivalent to $H^{A,V}$, more precisely,
\[H^{A+\nabla\Lambda,V}=e^{i\Lambda}H^{A,V}e^{-i\Lambda}.\]
This property is usually referred to as \emph{gauge invariance}, cf. \cite{CTT10b, L83}. We observe a similar behaviour in our case.

\begin{theorem}\label{T:gauge}
Assume $V\in L_\infty(K,\nu)$ and $a,b\in\mathcal{H}_\infty$. If $b=a+\partial\lambda$ with some $\lambda\in\mathcal{F}$, then $dom\:H^{b,V}=dom\:H^{a,V}$ and
\[H^{b,V}=e^{i\lambda}H^{a,V}e^ {-i\lambda}.\]
\end{theorem}

\begin{remark}
Recall Remark \ref{R:ortho}. Let $\mathcal{P}_\bot$ denote the orthogonal projection onto $(Im\:\partial)^\bot$. By Theorem \ref{T:gauge} the operator $H^{a,V}$ is uniquely determined up to unitary equivalence by $V$ and $\mathcal{P}_\bot a$. In this sense we can always restrict attention to divergence free vector potentials $a\in\mathcal{H}_\infty\cap ker\:\partial^\ast$. In the classical case the condition $\partial^\ast a=0$ is referred to as \emph{Coulomb gauge condition}, see for instance \cite{Si73}. 
\end{remark}

To prove Theorem \ref{T:gauge} we first verify the following lemma.

\begin{lemma}\label{L:gauge}
Let $a\in\mathcal{H}_\infty$, let $\lambda=\Lambda\circ h$, $\Lambda\in C^1(\mathbb{R}^2)$ and set $b:=a+\partial \lambda$. Then we have
\[\mathcal{E}^{a,0}(e^{-i\lambda} f)=\mathcal{E}^{b,0}(f),\ \ f\in\mathcal{F}_\mathbb{C}.\]
\end{lemma}
\begin{proof} Note first that under the stated hypotheses we also have $b\in\mathcal{H}_\infty$.

Assume that $f=F\circ h$ with $F\in C^1_\mathbb{C}(\mathbb{R}^2)$. By (\ref{E:fiberiso}) we have
\[\Phi_x(-i\partial(e^{-i\lambda} f)_x)=-iZ_x(-i\nabla\Lambda(y)e^{-i\Lambda(y)}F(y)+e^{-i\Lambda(y)}\nabla F(y))=\Phi_x((-e^{-i\lambda}\partial\lambda-ie^{-i\lambda}\partial f)_x)\]
for $\nu$-a.a. $x\in X$, and as the $\Phi_x$ are isomorphisms also
\[-i\partial(e^{-i\lambda} f)_x=(-e^{-i\lambda}\partial\lambda-ie^{-i\lambda}\partial f)_x\]
in the spaces $\mathcal{H}_{\mathbb{C},x}$. Integrating, we obtain
\[-i\partial(e^{-i\lambda} f)=-e^{-i\lambda}\partial\lambda-ie^{-i\lambda}\partial f\]
in $\mathcal{H}_\mathbb{C}$ and therefore
\[e^{i\lambda}(-i\partial -a)e^{-i\lambda}f=(-i\partial-b)f,\]
what implies the desired equality. For general $f\in\mathcal{F}_\mathbb{C}$ let $(f_n)_n$ with $f_n=F_n\circ h$ and $F_n\in C^1_\mathbb{C}(\mathbb{R}^2)$ be a sequence approximating $f$ in $\mathcal{E}$. Then
\[|\mathcal{E}^{a,0}(f)-\mathcal{E}^{a,0}(f_n)|=|\mathcal{E}(f)-\mathcal{E}(f_n)|+2|\left\langle \partial f,a f\right\rangle_\mathcal{H}-\left\langle \partial f_n,a f_n\right\rangle_\mathcal{H}|+|\left\|af\right\|_\mathcal{H}^2-\left\|a f_n\right\|_\mathcal{H}^2|.\]
The first term clearly tends to zero, for the second we have the estimate 
\begin{multline}
2|\left\langle \partial f-\partial f_n,af\right\rangle_\mathcal{H}|+2|\left\langle \partial f_n, af-af_n\right\rangle_\mathcal{H}|\notag\\
\leq 2\mathcal{E}(f-f_n)^{1/2}\left\|af\right\|_\mathcal{H}+2\left\|f-f_n\right\|_{L_\infty(K,\nu)}\left\|a\right\|_\mathcal{H}\sup_n\mathcal{E}(f)^{1/2},
\end{multline}
which tends to zero by (\ref{E:resistest}). The third does not exceed
\[\left\|a(f-f_n)\right\|_\mathcal{H}(\left\|af\right\|_\mathcal{H}+\left\|af_n\right\|_\mathcal{H}),\]
what is similarly seen to converge to zero.
\end{proof}

We prove Theorem \ref{T:gauge}.

\begin{proof}
For $\lambda $ as in the theorem we obviously have $\partial\lambda\in\mathcal{H}_\infty$. Let $(a_n)\subset \mathcal{S}$ be a sequence approximating $a$ in $\mathcal{H}$ and let $(\lambda_n)_n$ with $\lambda_n=\Lambda_n\circ h$ and $\Lambda_n\in C^1(\mathbb{R}^2)$ be a sequence approximating $\lambda$ in $\mathcal{E}$. Set $b_n:=a_n+\partial\lambda_n$. Then Lemma \ref{L:gauge} implies
\begin{equation}\label{E:simplecase}
\mathcal{E}^{a_n,V}(e^{-i\lambda_n} f)=\mathcal{E}^{b_n,V}(f), \ \ f\in\mathcal{F}_\mathbb{C}.
\end{equation}
Now we first claim that
\begin{equation}\label{E:approxbn}
\mathcal{E}^{b,V}(f)=\lim_n\mathcal{E}^{b_n,V}(f),\ \ f\in\mathcal{F}_\mathbb{C}.
\end{equation}
This follows from
\[|\mathcal{E}^{b,V}(f)-\mathcal{E}^{b_n,V}(f)|\leq 2|\left\langle \partial f, (b-b_n)f\right\rangle_\mathcal{H}|+|\left\|fb\right\|_\mathcal{H}^2-\left\|fb_n\right\|_\mathcal{H}^2|\]
because we have the upper bound
\[2\mathcal{E}(f)^{1/2}\left\|f\right\|_{L_\infty(K,\nu)}\left\|b-b_n\right\|_\mathcal{H}\leq 2\mathcal{E}(f)\left\|b-b_n\right\|_\mathcal{H}\]
for the first summand, and for the second,
\[\left\|(b-b_n)\right\|_\mathcal{H}\left\|f\right\|_{L_\infty(K,\nu)}(\left\|fb\right\|_\mathcal{H}+\left\|fb_n\right\|_\mathcal{H}).\]
Combining, we arrive at (\ref{E:approxbn}). Next, note that we also have 
\begin{equation}\label{E:approxan}
\mathcal{E}^{a,V}(e^{-i\lambda} f)=\lim_n\mathcal{E}^{a_n,V}(e^{-i\lambda_n}f),\ \ f\in\mathcal{F}_\mathbb{C}.
\end{equation}
To see this, note that 
\begin{align}
|\mathcal{E}^{a,V}(e^{-i\lambda}f)&-\mathcal{E}^{a_n,V}(e^{-i\lambda_n}f)|\notag\\
&\leq |\mathcal{E}^{a,V}(e^{-i\lambda} f)-\mathcal{E}^{a,V}(e^{-i\lambda_n} f)|+|\mathcal{E}^{a,V}(e^{-i\lambda_n} f)-\mathcal{E}^{a_n,V}(e^{-i\lambda_n} f)|\notag\\
&\leq 2|\left\langle \partial ((e^{-i\lambda}-e^{-i\lambda_n})f), a\right\rangle_\mathcal{H}|+|\left\|a e^{-i\lambda}f\right\|_\mathcal{H}^2-\left\|a e^{-i\lambda_n} f\right\|_\mathcal{H}^2|\notag\\    
&\ \ + 2 |\left\langle \partial(e^{-i\lambda_n}f), a-a_n\right\rangle_\mathcal{H}|+|\left\| a e^{-i\lambda_n}f\right\|_\mathcal{H}^2-\left\|a_n e^{-i\lambda_n} f\right\|_\mathcal{H}^2|,\notag
\end{align}
what can be estimated by
\begin{align}\label{E:summandstobecovered}
2&\left\|\partial(1-e^{i(\lambda-\lambda_n)})e^{-i\lambda}f\right\|_\mathcal{H}+2\left\|(1-e^{i(\lambda-\lambda_n)})\partial(e^{-i\lambda} f)\right\|_\mathcal{H}\\
&+\left\|(1-e^{i(\lambda-\lambda_n)})e^{-i\lambda} fa\right\|_\mathcal{H}\left(\left\|e^{-i\lambda}fa\right\|_\mathcal{H}+\left\|e^{-i\lambda_n}fa\right\|_\mathcal{H}\right)\notag\\
&+2\left\|\partial(e^{-i\lambda_n}f)\right\|_\mathcal{H}\left\| a-a_n\right\|_\mathcal{H}\notag\\
&+\left\|e^{-i\lambda_n}f(a-a_n)\right\|_\mathcal{H}\left(\left\|e^{-i\lambda_n}fa\right\|_\mathcal{H}+\left\|e^{-i\lambda_n}fa\right\|_\mathcal{H}\right).\notag
\end{align}
By the chain rule, \cite[Theorem 3.2.2]{FOT94}, we have 
\[\mathcal{E}(1-e^{i(\lambda-\lambda_n)})=\int_K|e^{i(\lambda-\lambda_n)}|^2\Gamma(\lambda-\lambda_n)d\nu=\mathcal{E}(\lambda-\lambda_n),\]
and consequently
\[\lim_n\int_K|e^{i\lambda}f|^2\Gamma(1-e^{i(\lambda-\lambda_n)})d\nu\leq \lim_n\left\|f\right\|_{L_\infty(K,\nu)}\mathcal{E}(\lambda-\lambda_n)=0.\]
The function $z\mapsto 1-e^{-iz}$ is continous and bounded, and by (\ref{E:resistest}) we have $\lim_n \lambda_n=\lambda$ uniformly on $K$. Therefore
\[\lim_n\int_K|1-e^{i(\lambda-\lambda_n)}|^2\Gamma(e^{-i\lambda}f)d\nu=0\]
by bounded convergence. Combining these estimates we see that the first line in (\ref{E:summandstobecovered}) converges to zero. The other terms in (\ref{E:summandstobecovered}) obey similar bounds, use (\ref{E:resistest}), (\ref{E:pointwisemult}) and note that
\[\sup_n\mathcal{E}(e^{-i\lambda_n})=\sup_n\int_K|e^{-i\lambda_n}|^2\Gamma(\lambda_n)d\nu=\sup_n\mathcal{E}(\lambda_n)<\infty.\]
They all converge to zero, and (\ref{E:approxan}) becomes evident. Clipping (\ref{E:simplecase}), (\ref{E:approxbn}) and (\ref{E:approxan}) we obtain the equality of closed forms
\[\mathcal{E}^{a,V}(e^{-i\lambda} f)=\mathcal{E}^{b,V}(f), \ \ f\in\mathcal{F}_\mathbb{C},\]
which implies the coincidence of the associated self-adjoint operators and therefore Theorem \ref{T:gauge}.
\end{proof}

\section{Other fractals}\label{S:other}

The results of the preceding sections easily carry over to regular resistance forms on finitely ramified fractals, see e.g. \cite{IRT, Ki03, Ki12, T08} for background and precise definitions. In this situation it is always possible to find a complete (up to constants) energy orthonormal system $h_1,...,h_k$ of harmonic functions. We assume that 
\[h(x):=(h_1(x),...,h_k(x)), \ \ x\in X,\]
defines a homeomorphism from $X$ onto its image $h(X)$ in $\mathbb{R}^k$. We can then define a (normed) Kusuoka energy measure as the sum
\[\nu:=\nu_{h_1}+...+\nu_{h_k}\]
of the corresponding energy measures, cf. \cite[Definition 3.5]{T08}. By \cite[Theorem 3]{T08} the form $(\mathcal{E},\mathcal{F})$ is a regular Dirichlet form on $L_2(X,\nu)$. Following Kusuoka \cite{Ku89} one can construct a matrix valued function $Z$ on $X$ as before, such that $Z_{ij}$ is a version of the density of $\nu_{h_i,h_j}$ with respect to $\nu$. By \cite[Theorem 6]{T08} the collection of functions $F\circ h$ with $F\in C_\mathbb{C}^1(\mathbb{R}^k)$ is dense in $\mathcal{F}_\mathbb{C}$ and formula (\ref{E:Kigami}) still holds. All results of Section \ref{S:vector} may be rewritten for $X$ in place of the Sierpinski gasket $K$. The constructions of Sections \ref{S:Dirac} and \ref{S:magnetic} do not depend on the specific structure of $X$ and therefore remain valid as well.
\begin{remark}
We would like to  point out that by minor modifications Section \ref{S:magnetic} applies to any local regular Dirichlet form that admits energy densities (i.e. a carr\'e du champ) with respect to the given reference measure (see \cite{BH91,HRT}). Section \ref{S:Dirac} does not require energy densities, although it is most naturally applicable for  topologically one-dimensional spaces of arbitrary large spectral and Hausdorff dimensions (see \cite{HT}). 
\end{remark}

\end{document}